\let\oldvec\vec
\documentclass{article}
\let\vec\oldvec
\usepackage{amssymb}
\usepackage{amsmath}
\usepackage{amsthm}
\usepackage{graphicx}
\usepackage{nicefrac}
\usepackage{subfigure}

\newtheorem{theorem}{Theorem}[section]

\newtheorem{corollary}[theorem]{Corollary}

\title{Budget-restricted utility games with ordered strategic decisions\,\thanks{This work was partially supported by the German Research Foundation (DFG) within the Collaborative Research Centre ``On-The-Fly Computing'' (SFB 901) and by the EU within FET project MULTIPLEX under contract no.\ 317532.}\,
\thanks{An extended abstract of this paper has been accepted for publication in the proceedings of the 7th International Symposium on Algorithmic Game Theory (SAGT), available at www.springerlink.com}}
\author{Maximilian Drees \and S\"oren Riechers \and Alexander Skopalik}

\date{
    Heinz Nixdorf Institute \& Computer Science Department\\[0.2em]
    University of Paderborn (Germany)\\[0.2em]
    F\"urstenallee 11, 33102 Paderborn\\[0.2em]
}

\begin{document}

\maketitle

\begin{abstract}
 We introduce the concept of \emph{budget games}.
 Players choose a set of tasks and each task has a certain demand on every resource in the game.
 Each resource has a budget.
 If the budget is not enough to satisfy the sum of all demands, it has to be shared between the tasks.
 We study strategic budget games, where the budget is shared proportionally.
 We also consider a variant in which the order of the strategic decisions influences the distribution of the budgets.
 The complexity of the optimal solution as well as existence, complexity and quality of equilibria are analyzed.
 Finally, we show that the time an ordered budget game needs to convergence towards an equilibrium may be exponential.
\end{abstract}

\section{Introduction}
 
Recent advancements of network technology enabled and simplified outsourcing of processing and storing information to remote facilities. The offering of such services in a competitive environment has become known as cloud computing.
The competitive aspect is twofold. 
On the one hand, customers compete over the allocation of various types of services and resources like bandwidth, memory space, computing power etc. These resources are usually limited in capacity and as soon as the demand exceeds that capacity customers' demand can only be satisfied partially.
On the other hand, the service providers face strategic decisions in the markets which have to take into account the budget of their clients. As long as a client can afford all the desired products, his budget has no consequence. But once their total costs exceed his budget, he has to split it between them. When deciding to offer a product, a provider therefore has to consider the remaining budgets of the interested clients.

We study this in a game theoretic setting called {\em budget games} in which several tasks (or products) have a certain demand for resources (or money) and the resources (or clients) have budgets. As long as the sum of the demands does not exceed the budget, the demands can be completely satisfied, otherwise only partially.
For example in scheduling, where tasks are allocated to one or more servers. Each server disposes of a limited amount of computational capacity, space or bandwidth and when it runs too many tasks, this capacity has to be split between them. Naturally, every job will aim to obtain as much capacity as it needs, which may vary between the different servers. Also, not every server combination may be possible for every task. 


We study budget games as strategic games as well as in a variant that takes into account temporal aspect. 
Strategic games are often analyzed as one-shot games which do not capture situations like a new provider  entering a market having a disadvantage against those already established. The clients prioritize the products they already know and spend only what may be left of their budget on what the new provider offers. As a result, he cannot gain more than what is left of a clients budget.

In the strategic game the utility of a resource is shared proportional among all tasks. In the second approach, called {\em ordered} budget games, we also take into account the order in which the tasks arrived.
Each resource has an ordering of the tasks and its utility is allocated to the tasks in that order. If a player decides to deviate to another strategy, the tasks that are allocated to different resources are moved to the last position in the ordering of those resources.

%
\subsubsection{Related Work}
There are several models which share similarities to budget games. Li et al. \cite{Li05} developed cost-sharing mechanisms for set cover games. Every element $e_i$ has a coverage requirement $r_i$, every set $S_j$ has a cost $c_j$ and the multiplicity of $e_i$ in $S_j$ is $k_{j,i}$. The multiplicity states how many times $e_i$ is covered by $S_j$. The sets are chosen on the condition that $e_i$ has to be covered at least $r_i$ many times. The total costs are distributed between the elements such that the result is $\frac{1}{\ln(d_\text{max})}$-budget-balanced and fair under core. In \cite{Li10}, Li et al. analyze set cover games in which the elements are the agents and declare bids for the sets. They give mechanisms which decide which elements will be covered, which sets are used and how much each element is charged.

Other games have been defined on the facility location problem \cite{Jain01}. In \cite{Ahn04}, Ahn et al. studied the Voronoi game in which two players alternately choose their facilities and the space they control is determined by the nearest-neighbor rule. They give a winning strategy for player 2, although player 1 can ensure that the advantage is only arbitrarily small.

Also related to our model are congestion games. Rosenthal \cite{Rosenthal73} showed that they always have a pure Nash equilibrium. Milchtaich \cite{Milchtaich96} extended this result to weighted congestion games with player-specific payoff functions, where the utility of player $i$ playing strategy $j$ is a monotonically nonincreasing function $S_{ij}$ of the total weight of all players with the same strategy. Mavronicolas et al. \cite{Mavronicolas07} considered the special case of latency functions $f_{ie} = g_e \odot c_{ie}$, where $g_e$ is the latency function of resource $e$, $c_{ie}> 0$ and $\odot$ is the operation of an abelian group. A characterization of the class of congestion games with pure Nash equilibria was recently given by Gairing and Klimm \cite{Gairing}. They showed that the player-specific cost functions of the weighted players have to be affine transformations of each other as well as be affine or exponential. These games emphasize that the impact of the same strategic choice may vary between the players.

\noindent
Finally, the strategic version of our game is a basic utility game.
One property of basic utility games is that the social welfare function is submodular and non-decreasing, which
is used in Section \ref{budget:problemComplexity} to approximate the optimal solution for any of our games.
Vetta \cite{Vetta02} showed that any basic utility game has a Price of Anarchy at most 2.
We prove the same for the non-strategic ordered budget games.

\subsubsection{Our Contribution}

We show that computing an optimal allocation for both variants of budget games is {\sf NP}-hard in general but can be approximated within a factor of $1-\nicefrac{1}{e}$ if the strategies of the players have a matroid structure.

In standard budget games a stable solution, i.\,e., a pure Nash equilibrium, might not exist and deciding if one exists is {\sf NP}-hard.
For ordered budget games the situation is more positive. Nash and even strong equilibria exist and can be computed in polynomial time.
We show that this complexity result cannot be extended to super strong equilibria as these are {\sf NP}-hard to compute. 
Moreover, we compare the performance of equilibria to optimal solutions and show that the price of (strong) stability is $1$ and the price of (strong) anarchy is $2$.
Concerning the convergence of repeated improvement steps we show that the dynamic that emerges
if players repeatedly make improving moves converges towards a Nash equilibrium and this is even true for simultaneous moves of several players if ties are broken in a certain way.
However, there are games and initial strategy profiles in which the convergence process may take exponentially long.

\section{Model}
\label{model}
A \emph{budget game} $\mathcal{B}$ is a tuple $(\mathcal{N}, \mathcal{R}, (b_r)_{r \in \mathcal{R}}, (\mathcal{S}_i)_{i \in \mathcal{N}}, (u_i)_{i \in \mathcal{N}})$, where the set of players is denoted by $\mathcal{N} = \{1,\ldots,n\}$, the set of resources by $\mathcal{R} = \{1,\ldots,m\}$, and the budget of resource $r$ by $b_r$. Each player $i$ has a set of tasks $\mathcal{T}_i = \{ t^i_1,\ldots,t^i_{q_i} \}$ with $t^i_{k} \in \mathbb{R}^m_{\geq 0}$.
 For a task $t \in \mathcal{T}_i$, we use $t(r)$ to denote the demand for resource $r$.
 We say a task $t$ is \textit{connected} to a resource $r$ if $t(r) > 0$.
 If the task demands the full resource, i.\,e. $t(r) = b_r$, we say that $t$ is \textit{fully connected} to $r$.
Now, let $\mathcal{T} = \cup_{i \in \mathcal{N}} \mathcal{T}_i$ denote the set of all tasks.
A strategy of a player is a set of tasks and $\mathcal{S}_i \subseteq 2^{\mathcal{T}_i}$ denotes the set of strategies available to player $i$. $\mathcal{S} = \mathcal{S}_1 \times \ldots \times \mathcal{S}_n$ is the set of strategy profiles and $u_i: \mathcal{S} \rightarrow \mathbb{R}_{\geq 0}$ denotes the private utility function player $i$ strives to maximize.
For a strategy profile $s = (s_1,\ldots,s_n)$, let $u_{t,r}(s): \mathcal{S} \rightarrow \mathbb{R}_{\geq 0}$ denote the utility of $t$ from $r$ and $u_i(s) := \sum_{t \in s_i} \sum_{r \in \mathcal{R}} u_{t,r}(s)$. We demand that the utilities are always valid, i.\,e. $\sum_{i \in \mathcal{N}} \sum_{t \in s_i}  u_{t,r} \le b_r$ for every $r \in \mathcal{R}$.


 We consider two different utility distribution rules and call the games \emph{standard budget games} (or simply budget games) and \emph{ordered budget games}.
 In a standard budget game, the utility of task $t \in s_i$ from resource $r$ is defined as
$u_{t,r}(s) := \min\left(t(r), \nicefrac{b_r \cdot t(r)}{\left(\sum_{j \in \mathcal{N}} \sum_{t' \in s_{j}} t'(r)\right)} \right).$

In an ordered budget game, the utilities do not only depend on the current strategy profile, but also on the course of the game up to this point. To that end a strategy profile is augmented by an ordering of the tasks for each resource.
 Let $\prec=(\prec_e)_{r \in \mathcal{R}}$ be a vector of total orders on the set $\mathcal{T}$. The utility of a task $t \in s_i$ in $(s,\prec)$ is  $u_{t,r}(s,\prec) :=  t(r)$ if $\sum_{j \in \mathcal{N}} \sum_{t' \in s_j \text{ with } t' \prec_r t} t'(r) \leq b_{r}$ and $u_{t,r}(s,\prec) := \max\left(0, b_r - \sum_{j \in \mathcal{N}} \sum_{t' \in s_j \text{ with } t' \prec_r t} t'(r)\right)$ otherwise.

%

When player $i$ changes its strategy from $s_i$ to $s'_i$ all new tasks are 
moved to the end of $\prec_r$ for all resources.
Let $\tau = s'_i \setminus s_i$ then the new state is $((s'_i,s_{-i}),\prec')$ with $x \prec'_r y$ if and only if $x \prec_r y$ and $x \prec'_r t$  for all $x,y \in \mathcal{T} \setminus \tau$ and $t \in \tau$.
Here, the order for given tasks of the same player is arbitrary, as it does not change the utility function of the specific player.
For strategy changes of a coalition $C \subseteq \mathcal{N}$ of players the definition is analogous and we set $\tau = \cup_{i \in C} (s'_i \setminus s_i$).
For the ordering between tasks in $\tau$, we show two tie-breaking rules in Section \ref{budget:orderedgames}.


A pure Nash equilibrium (NE) is a strategy profile $s$ in which no player has an incentive to deviate, i.\,e., there is no  $s'_i \in \mathcal{S}_i$ such that $u_i(s'_i,s_{-i}) > u_i(s)$ for all $i \in \mathcal{N}_i$. A strong equilibrium is a profile $s$ in which there is no coalition $C\subseteq \mathcal{N}$ which can improve, i.\,e., there is no $s'_C \in \times_{i\in C} \mathcal{S}_i$ such that $u_i(s'_C,u_{-C}) >u_i(s_C,s_{-C})$ for all $i \in \mathcal{N}$. For super strong equilibrium we only demand that this inequality is strict for at least one player.

For a strategy profile $s$, $u(s) := \sum_{i \in \mathcal{N}} u_i(s)$ is the social welfare of $s$. The optimal solution of $\mathcal{B}$ is the strategy profile $opt$ with $u(opt) \geq u(s)$ for every $s \in \mathcal{S}$. The price of anarchy (PoA) is defined as $\max \frac{u(opt)}{u(s)}$, the price of stability (PoS) as $\min \frac{u(opt)}{u(s)}$, where $s$ is a NE. Analogously the price of (super) strong anarchy and stability is defined with $s$ being a (super) strong equilibrium. 

\section{Complexity of the Optimal Solution}
\label{budget:problemComplexity}

For any form of budget game, the social welfare is independent of the order of the tasks. The following results hold for both standard and ordered budget games.

\begin{theorem}
 Computing the optimal solution for a budget game with respect to social welfare is {\sf NP}-hard,
 even if the tasks and strategy sets of all players are equal and the strategies are restricted to singletons.
\end{theorem}

\begin{proof}
 We give a reduction from the maximum set coverage problem.
 An instance $\mathcal{I} = (\mathcal{U},\mathcal{W},w)$ of this problem is given by a set $\mathcal{U}$,
 a collection of subsets $\mathcal{W} = \{\mathcal{W}_1,\ldots,\mathcal{W}_q\}$ with $\mathcal{W}_i\subseteq \mathcal{U}$
 and an integer $w \in \mathbb{N}$.
 The task is to cover as many elements from $\mathcal{U}$ as possible by choosing at most $w$ sets from $\mathcal{W}$.
 
 From $\mathcal{I}$, we create a budget game $\mathcal{B} = (\mathcal{N}, \mathcal{R}, (b_r)_{r \in \mathcal{R}}, (\mathcal{S}_i)_{i \in \mathcal{N}}, (u_i)_{i \in \mathcal{N}})$.
 We create a number of $w$ players, that is $\mathcal{N} = \{1,\ldots,w\}$.
 Now, let the set of resources correspond to the set $\mathcal{U}$, i.\,e. $\mathcal{R} = \mathcal{U}$,
 and set the budget of each resource $j \in \mathcal{R}$ to $b_j = 1$.
 For each player $i$, we define the set of tasks as $\mathcal{T}_i = \{t_{\mathcal{W}_1},\ldots,t_{\mathcal{W}_q}\}$, 
 where the demands of a task are set to $t_{\mathcal{W}_k}(r) = 1$ for $r \in \mathcal{W}_k$ and $t_{\mathcal{W}_k}(r) = 0$ otherwise.
 Note that the set of tasks is equal for all players.
 Finally, we set the strategy space to be $\mathcal{S}_i = \{ \{t_{\mathcal{W}_k}\} \ | \ 1 \leq k \leq q \}$ for all $i \in \mathcal{N}$.
 
 Given a strategy profile $s$ for $\mathcal{B}$, the social welfare increases by 1 for every resource $r$ that is used by some task.
 This applies if and only if there is a set $\mathcal{W}_k$ with $r \in \mathcal{W}_k$ so that the chosen strategy
 of some player $i$ is $s_i = \{t_{\mathcal{W}_k}\}$.
 Choosing strategies for all players corresponds to choosing $w$ sets from $\{\mathcal{W}_1,\ldots,\mathcal{W}_q\}$ and thus
 a strategy profile for $\mathcal{B}$ also describes a solution for $\mathcal{I}$
 where the number of covered elements equals the social welfare of $s$.
 In addition, every solution for $\mathcal{I}$ can be transformed into a strategy profile for $\mathcal{B}$
 by assigning each chosen set $\mathcal{W}_k$ to one player $i$ by setting $s_i = \{t_{\mathcal{W}_k}\}$. 
 Again, the social welfare and the number of covered elements are equal.
 Therefore, the problems of finding an optimal solution for $\mathcal{B}$ and finding an optimal solution for $\mathcal{I}$ is equivalent. 
\end{proof}

 If the sets of strategies $\mathcal{S}_i$ correspond to bases of some matroid (with the tasks as elements), the optimal solution for a budget game can be approximated up to a constant factor, since computing an optimal solution corresponds to maximization of a submodular monotone function. A function $g: 2^\mathcal{U} \rightarrow \mathbb{R}$ over a set $\mathcal{U}$ is submodular if $g(X \cup \{u\}) - g(X) \geq g(Y \cup \{u\}) - g(Y)$ for $X \subseteq Y, u \notin Y$ and monotone if $g(A) \leq g(B)$ for all $A \subseteq B$. For budget games, the function mapping the set of tasks chosen by the players to the social welfare has these properties. Nemhauser et al. \cite{Nem} proved that greedy maximization yields an approximation factor of $1 - \frac{1}{e}$. In our case, this means always picking the task (out of all) with the highest utility next. The resulting strategies are then valid, provided the number of tasks in each is not too large. Feige \cite{Feige98} showed that there is no better approximation algorithm for the maximum set coverage problem unless ${\sf P = NP}$. Therefore, we conclude the following result.
\begin{corollary}
 In a matroid budget game, greedy maximization of the social welfare creates a strategy profile $s$ with $\frac{u(opt)}{u(s)} \leq 1 - \frac{1}{e}$. This bound is tight if ${\sf P \neq NP}$.
 \label{lemma:greedyMax}
\end{corollary}

\section{Standard Budget Games}
A (standard) budget game does not always possess a NE. In addition, the question whether a given game instance has at least one NE is {\sf NP}-hard. 

\begin{theorem}
 To decide for a budget game $\mathcal{B}$ whether it has a NE is {\sf NP}-complete.
\end{theorem}

\begin{proof}
 Given $\mathcal{B}$ and a strategy profile $s$, we can verify if $s$ is a NE in polynomial time.
 Therefore, the problem is in {\sf NP}.
 We prove that it is {\sf NP}-hard by reduction from the exact cover by 3-sets problem.
 Given an instance $\mathcal{I} = (\mathcal{U},\mathcal{W})$ consisting of a set $\mathcal{U}$ with $|\mathcal{U}| = 3m$ and 
 a collection of subsets $\mathcal{W} = \mathcal{W}_1,\ldots,\mathcal{W}_q \subseteq \mathcal{U}$ with $|\mathcal{W}_k| = 3$ for every $k$, the question whether $\mathcal{W}$ contains an exact cover for $\mathcal{U}$ in which every element is covered by exactly one subset is {\sf NP}-hard.
 We create a budget game $\mathcal{B} = (\mathcal{N}, \mathcal{R}, (b_r)_{r \in \mathcal{R}}, (\mathcal{S}_i)_{i \in \mathcal{N}}, (u_i)_{i \in \mathcal{N}})$ as follows.
 $\mathcal{N} = \{1,\ldots,q,A,B,C,D\}$ with $\mathcal{T}_i = \{t^i_0,t^i_1\}$ and $\mathcal{S}_i = \{\{t^i_0\},\{t^i_1\}\}$ for every $i = 1,\ldots,q,A,B,C,D$.
 The players $1,\ldots,q$ correspond to the sets in $\mathcal{W}$.
 We introduce the following resources and budgets. The actual values for $\gamma$ and $\delta$ will be defined later on.
 \begin{center}
 \begin{tabular}{r|c|c|c|c|c|c|c|c|c|c|c}
 Resource \  & \ $r_j, j \in \mathcal{U}$ \ & \ $r_{e,i}, i \in [q]$ \ & \ $r_f$ & \ $r_{aux}$ \ & \ $r'_{aux}$ \  & \ $r'_1$ \ & \ $r'_2$ \ & \ $r'_3$ \ & \ $r'_4$ \ & \ $r'_5$ \ & \ $r'_6$ \ \\ \hline
 Budget \  & 1 & \ $\nicefrac{8}{3}$ \ & \ 100 \ & \ 100 \ & \ $\gamma$ & \ 5 \ & \ 10 \ & \ 10 \ & \ 5 \ & \ 10 \ & \ 15 \ \\
 \end{tabular}
 \end{center}
 
 Finally, we list all demands which are not 0. A sketch of the resulting game can be found in Figure \ref{fig:NPhardNE}. 
\begin{center}
 \begin{tabular}{r||c||c||c|c|c}
 Task\,&\,$t^i_0(r_j)\,\forall i\in\{1,\ldots,q\},j\in \mathcal{W}_i$\,&\,$t^i_1(r_{e,i})$\,&\,$t^A_0(r'_1)$\,&\,$t^A_0(r'_2)$\,&\,$t^A_0(r_{aux})$\,\\ \hline
 Demand\,&$1$&\,$997\nicefrac{1}{3}$\,&\,$5$\,&\,$10$\,&\,$33\nicefrac{1}{3}$\,
 \end{tabular}
 \end{center}
 \vspace{-3.5mm}
 \begin{center}
 \begin{tabular}{r||c||c|c|c||c||c|c|c}
 Task\,&\,$t^A_1(r'_3)$\,&\,$t^B_0(r'_3)$\,&\,$t^B_0(r'_4)$\,&\,$t^B_0(r_{aux})$\,&\,$t^B_1(r'_5)$\,&\,$t^C_0(r'_2)$\,&\,$t^C_0(r'_5)$\,&\,$t^C_0(r_{aux})$\,\,\\ \hline
 Demand\,&\,$990$\,&\,$10$\,&\,$5$\,&\,$33\nicefrac{1}{3}$\,&\,$990$\,&\,$990$\,&\,$10$\,&\,$33\nicefrac{1}{3}$\,
 \end{tabular}
 \end{center}
 \vspace{-4mm}
 \begin{center}
 \begin{tabular}{r||c||c|c||c|c}
 Task\,&\,$t^C_1(r'_6)$\,&\,$t^D_0(r_{aux})$\,&\,$t^D_0(r'_{aux})$\,&\,$t^D_1(r_e)$\,&\,$t^D_1(r_f)$\,\\ \hline
 Demand\,&\,$11$\,&\,$\delta$\,&\,$\gamma$\,&\,$\nicefrac{8}{3}\cdot q$\,&\,$100$\,
 \end{tabular}
 \end{center} 
 Basically, our game consists of two smaller ones. The first involves the players $1,\ldots,q$ and is based on $\mathcal{I}$, the other revolves around $A,B$ and $C$ and is mostly constant. Player $D$ forms a connection between the two games. The fact whether $\mathcal{I}$ has a solution determines how the NE in the first game looks like. If $\mathcal{I}$ can be solved, there is a NE which causes $D$ to participate in the first game. This in turn is necessary for the existence of any NE in the second game and therefore for the existence in $\mathcal{B}$ as a whole. We start by analyzing the first game.
 
 Each player $i = 1,\ldots,q$ has the decision between $t^i_0$ and $t^i_1$. Choosing $t^i_0$ corresponds to picking the set $\mathcal{W}_i$ as part of a solution for $\mathcal{I}$. If the three resources connected to $t^i_0$ are not covered by any other task, the utility of player $i$ is 3. Otherwise, it is at most $\frac{5}{2} = 1 + 1 + \frac{1}{2}$. This is already the case when one of the three resources is covered by only one other task. Since the utility of $t^i_1$ is always greater than $\frac{5}{2}$, a player has no incentive to choose $t^i_0$ unless he receives the full budget of the three connected resources. If there is an exact cover in $\mathcal{I}$, it consists of $m$ sets. This leaves $q-m$ players to pick the task $t^i_1$. On the other hand, if $m$ players choose the task $t^i_0$, then they represent an exact cover.
 
 \begin{figure}[!ht]
\centering
  \includegraphics[width=0.5\textwidth]{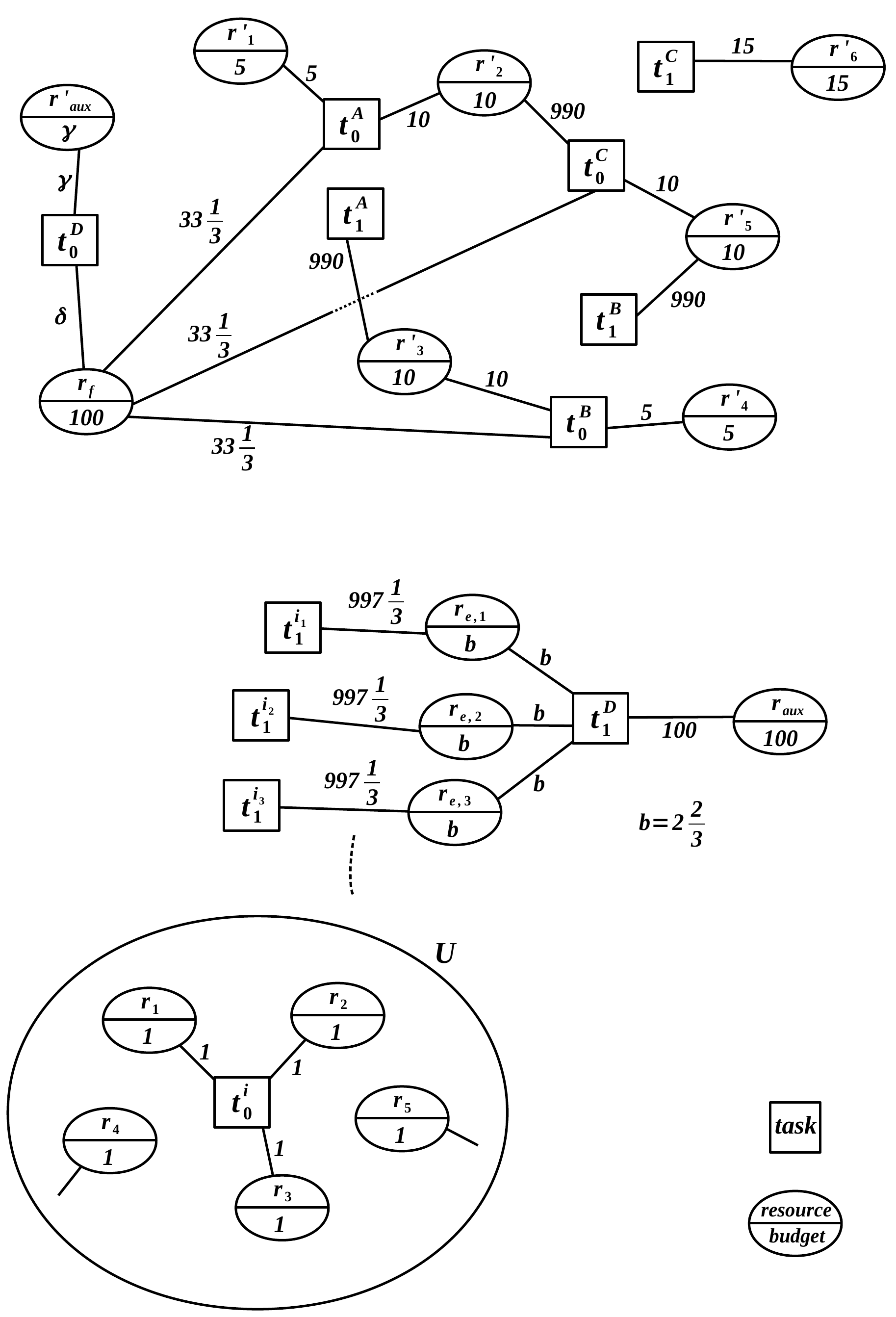}
  \caption{A standard budget game, assembled from two smaller ones, for which it is {\sf NP}-hard to determine whether it has a Nash equilibrium. The upper game has a repeating cycle in which the players $A,B$ and $C$ keep switching between their two strategies as long as $D$ plays $\{t^D_0\}$. He only changes his strategy to $\{t^D_1\}$ if no more than $q-m$ players $i$ of the lower game play $\{t^i_1\}$, which only happens when the remaining $m$ players form an exact cover of the resources in $U$.}
  \label{fig:NPhardNE}  
\end{figure}
 
 For player $D$, the utility of $t^D_1$ is 100 plus his share of the budgets of $r_{e,i}$. For $q-m$ players picking $t^i_1$, this share is $\alpha := (q-m) \frac{b^2}{1000} + mb$.
 For one additional player, the share decreases to $\beta := (q-m+1) \frac{b^2}{1000} + (m-1)b$.
 We pick $\gamma$ such that it has the property $\alpha > \gamma > \beta$. Then we choose $\delta > \frac{100(\beta - \gamma + 100)}{\gamma - \beta}$. One can verify that with these values, player $D$ will only pick the task $t^D_1$ if he has to share the resource $r_e$ with $q-m$ other players. As mentioned before, this is equivalent to $\mathcal{I}$ having an exact cover.
 
 Now consider the other half of our game. If player $D$ picks the task $t^D_0$, his demand on the resource $r_f$ is so high that the tasks $t^A_0$, $t^B_0$ and $t^C_0$ hardly profit from $r_f$. Intuitively speaking, for the players $A,B$ and $C$, the resource $r_f$ \textit{almost} does not exist. This prevents any stable state, as there is always one of the three players who can improve its utility by switching its strategy. We conclude that the existence of a Nash equilibrium is equivalent to $\mathcal{I}$ having an exact cover. 
\end{proof}

%
As a finite strategic game, every budget game has a mixed Nash equilibrium.
It is also a basic utility game and from \cite{Vetta02} we know that the price of anarchy is at most 2 for this class of games. 
If a budget game has a NE, this upper bound applies as well. We can get arbitrarily close to it as shown in the following example.
Let $\mathcal{B}$ be a budget game with $\mathcal{N} = \{1,\ldots,n+1\}$, $\mathcal{T}_i = \{t^i_0,t^i_1\}$ for $i=1,\ldots,n$ and $\mathcal{T}_{n+1} = \{t^{n+1}\}$.
Each player may only choose a single task, i.\,e. $\mathcal{S}_i = \{\{t^i_0\},\{t^i_1\}\}$ and $\mathcal{S}_{n+1} = \{\{t^{n+1}\}\}$. There are two resources $\mathcal{R} = \{r_1,r_2\}$ with $b_1 = b_2 = 1$.
The demands are $t^i_0(r_1) = \frac{1}{n+1}-\varepsilon$, $t^i_1(r_2) = b$, $t^n(r_2) = b$ and 0 else.
The optimal solution is the strategy profile $opt = (t^1_0,\ldots,t^n_0,t^{n+1})$ with a social welfare of $n \cdot (\frac{1}{n+1} - \varepsilon) + 1$. The only NE is $s = (t^1_1,\ldots,t^{n}_1,t^{n+1})$ with a social welfare of 1.

%

%
\section{Ordered Budget Games}
\label{budget:orderedgames}
We now turn to an extension of budget games namely ordered budget games that take into account chronological aspects. 
Note that ordered budget games are not strategic games as the utility of a player does not only depend on the strategy profile but also on the {\em order} in which they made their choices. 
For ordered budget games, the social welfare function is a potential function. Since every strategy change by a player (or a coalition of players) does not decrease the utility of the remaining players, it is easy to observe that
every improvement step by a player (or a coalition of players where every player improves his utility) increases social welfare.

\begin{corollary}
The social welfare function is a potential function as every improvement step of a player increases social welfare. 
\end{corollary}
Using this insight we derive a simple method to compute a strong equilibrium.
\begin{theorem}
\label{theo:computeEq}
A strong equilibrium can be computed in time $\mathcal{O}(n)$.
\end{theorem} 
\begin{proof}
A (strong) equilibrium can be computed in time $\mathcal{O}(n)$ by inserting players one after the other.
In the resulting state, no player has an incentive to deviate from its strategy as long as the players which have been inserted before him play the strategy they chose when they were inserted.
\end{proof}

Thus, computing both Nash and strong equilibria can be done in polynomial time.
However, if we consider super strong equilibria, the situation is different.
We show that finding such a state is {\sf NP}-hard.
\begin{theorem} 
 Computing a super-strong equilibrium for an ordered budget game with $n$ players is {\sf NP}-hard, even if the number of strategies per player is constant.
\end{theorem}
\begin{proof}
 We prove the theorem via a reduction from the monotone One-In-Three 3SAT problem. Given is a set $U = \{x_1,\ldots,x_n\}$ of variables and a collection $C$ of clauses over $U$ with $|c| = 3$ for each $c \in C$. In this context, monotone implies that no $c$ contains a negated literal. We therefore call the literals just variables.

 
 We construct an ordered budget game $ \mathcal{B} = (\mathcal{N}, \mathcal{R}, (b_r)_{r \in \mathcal{R}}, (\mathcal{S}_i)_{i \in \mathcal{N}}, (u_i)_{i \in \mathcal{N}})$ from the sets $U$ and $C$. Every variable $x_i \in U$ defines a player $i \in \mathcal{N}$ with $\mathcal{T}_i = \{ 0_i,1_i \}$. Every clause $c_j \in C$ defines two resources $r_{j,0},r_{j,1} \in \mathcal{R}$ with $b_{j,0} = 2$ and $b_{j,1} = 1$. $\mathcal{S}_i = \mathcal{T}_i$ for every player $i$. The demands are defined as
 $$ 0_i(r_{j,0}) = \left\{ \begin{array}{cl} 1, & \mbox{  if } x_i \in c_j  \\ 0, & \mbox{  else} \end{array} \right. \hspace{1cm} 1_i(r_{j,1}) = \left\{ \begin{array}{cl} 1, & \mbox{  if } x_i \in c_j  \\ 0, & \mbox{  else} \end{array} \right. $$  
 Set the remaining demands $0_i(r_{j,1})$ and $1_i(r_{j,0})$ to 0. Let $k_i$ be the number of clauses the variable $x_i$ occurs in. Then each task of $i$ has a demand of 1 on $k_i$ many resources and a demand of 0 on all others. The highest utility the player $i$ can obtain is also $k_i$. If there is a satisfying truth assignment $\phi$ for $C$, then each player can obtain this individual maximum. If $\phi(x_i) = 0$, let player $i$ choose strategy $0_i$, otherwise $1_i$. $\phi$ has to one-in-three property, which means that in each clause, only one variable is set to 1. Thus, every resource $r_{j,1}$ is covered by exactly one task $1_i$ and every resource $r_{j,0}$ by exactly two tasks $0_{i_1}$ and $0_{i_2}$. No resource experiences a demand higher than its budget, therefore the order of the tasks is not important here. In this case, the social welfare achieves a value of $\sum_{i \in \mathcal{N}} k_i$. If there exists a strategy profile with this social welfare in $\mathcal{B}$, then it induces in turn a satisfying truth assignment $\phi$ for $C$. Note that if such a strategy profile exists, it is also the only super-strong equilibrium of the game. In each other state, all players can form a coalition to collectively assume this strategy profile without reducing their utility. Therefore, computing a super-strong equilibrium for $\mathcal{B}$ determines whether $C$ can be satisfied or not.
\end{proof}

%

Since the optimal solution of an ordered budget game is a NE and even a super-strong equilibrium, we obtain the following bound on the price of (super strong) stability.
\begin{corollary}
The price of (super strong) stability of ordered budget games is $1$.
\end{corollary}

For the price of anarchy we obtain the following, nearly tight bound.

\begin{theorem}
 For every ordered budget game, the price of anarchy is at most $2$. For every $\varepsilon > 0$, there exists an ordered budget game with PoA $= 2 - \varepsilon$.  
 \label{lemma:PoAordered}
\end{theorem}

\begin{proof}
 We begin by upper bounding the price of anarchy of an ordered budget game $\mathcal{B}$. Let $(s,\prec)$ be a NE of $\mathcal{B}$ and $s^*$ be the strategy profile with the maximal social welfare. Note that the ordering of the players is irrelevant for the social welfare.
To simplify notation we will use $s$ and $(s_{-i},s^*_i)$ as a shorthand for $(s,\prec)$ and $((s_{-i},s^*_i),\prec')$ with $\prec'$ the new ordering as defined in Section~\ref{model}.
We can lower bound the social welfare of a NE $s$ as follows.
 \begin{align}
  \sum_{i \in \mathcal{N}} u_i(s) =& \sum_{r \in \mathcal{R}} \sum_{i \in \mathcal{N}} \sum_{t \in s_i} u_{t,r}(s)  \notag \\
                   \ge& \sum_{r \in \mathcal{R}} \sum_{i \in \mathcal{N}} \sum_{t \in s^*_i} u_{t,r}(s_{-i},s^*_i) \label{eq:Nash} \\
                   \ge& \sum_{r \in \mathcal{R}} \sum_{i \in \mathcal{N}} \sum_{t \in s^*_i} \min\left(t(r), b_r-\sum_{i' \ne i} \sum_{t' \in s_{i'}} u_{t',r}(s)\right) \label{eq:utilf} \\
                   \ge& \sum_{r \in \mathcal{R}} \sum_{i \in \mathcal{N}} \sum_{t \in s^*_i} \min\left(u_{t,r}(s^*), b_r-\sum_{i' \ne i} \sum_{t' \in s_{i'}} u_{t',r}(s)\right) \notag \\
                   \ge& \sum_{r \in \mathcal{R}_1} \sum_{i \in \mathcal{N}} \sum_{t \in s^*_i} \min\left(u_{t,r}(s^*), b_r-\sum_{i' \ne i} \sum_{t' \in s_{i'}} u_{t',r}(s)\right) + \sum_{r \in \mathcal{R}_2} \sum_{i \in \mathcal{N}} \sum_{t \in s^*_i} u_{t,r}(s^*) \label{eq:split}\\                    
                  \ge& \sum_{r \in \mathcal{R}_1} \left(b_r-\sum_{i' \in \mathcal{N}} \sum_{t' \in s_{i'}} u_{t',r}(s)\right)+ \sum_{r \in \mathcal{R}_2} \sum_{i \in \mathcal{N}} \sum_{t \in s^*_i} u_{t,r}(s^*) \notag \\
                  \ge& \sum_{r \in \mathcal{R}_1} \left( \sum_{i \in N} \sum_{t\in s^*_i} u_{t,r}(s^*) -\sum_{i \in \mathcal{N}} \sum_{t' \in s_{i}} u_{t',r}(s)\right)+ \sum_{r \in \mathcal{R}_2} \sum_{i \in \mathcal{N}} \sum_{t \in s^*_i} u_{t,r}(s^*) \notag \\
 \ge& \sum_{r \in \mathcal{R}}  \sum_{i \in N} \sum_{t\in s^*_i} u_{t,r}(s^*) -\sum_{r\in\mathcal{R}_1}\sum_{i \in \mathcal{N}} \sum_{t' \in s_{i}} u_{t',r}(s) \notag \\
                  \ge&  \sum _{i \in \mathcal{N}} u_i(s^*) - \sum_{r \in \mathcal{R}} \sum_{i \in \mathcal{N}} \sum_{t' \in s_{i}} u_{t',r}(s) \notag \\
                  \ge& \sum _{i \in \mathcal{N}} u_i(s^*) -  \sum_{i \in \mathcal{N}} u_i(s)  \label{eq:end} 
 \end{align}
 
 Observe that (\ref{eq:Nash}) follows from the Nash inequality and (\ref{eq:utilf}) from the definition of the utility functions. In (\ref{eq:split}) we partition $\mathcal{R}$ into $\mathcal{R}_1$ and $\mathcal{R}_2$ where $\mathcal{R}_1$ contains all resources with at least one task that evaluates the min statement to the second expression. That is there is a $i \in \mathcal{N}$ and a $t \in s^*_i$ with $u_{t,r}(s^*) > b_r-\sum_{i' \ne i} \sum_{t' \in s_{i'}} u_{t',r}(s)$. Adding $\sum_{i \in \mathcal{N}} u_i(s)$ to both sides at (\ref{eq:end}) shows that the price of anarchy is bounded by $2$. 
 
 For a lower bound, consider the game $\mathcal{B} = (\mathcal{N}, \mathcal{R}, (b_r)_{r \in \mathcal{R}}, (\mathcal{S}_i)_{i \in \mathcal{N}}, (u_i)_{i \in \mathcal{N}})$ with $\mathcal{N} = \{1,2\}$ with $\mathcal{T}_1 = \{t^1_1,t^1_2\}$ and $\mathcal{T}_2 = \{t^2\}$, $\mathcal{R} = \{r_1,r_2\}$ with $b_1 = b$ and $b_2 = b(1-\varepsilon)$. Set the demands to $t^1_1(r_1) = b$, $t^1_2(r_2) = b(1-\varepsilon)$, $t^2(r_2) = b$ and all others to 0. Set $\mathcal{S}_1 = \{ \{ t^1_1, \}, \{ t^1_2 \} \}$ and $\mathcal{S}_1 = \{ \{ t^2 \} \}$. In the optimal solution, the social welfare is $u_1(opt) + u_2(opt) = b - b \cdot \varepsilon + b = 2b - b \cdot \varepsilon$. If player 1 is inserted first, his best response is to open task $t^1_1$. This leads to a NE in which the utility of player 2 is 0 and the price of anarchy $2-\varepsilon$. $\mathcal{B}$ can be extended to $n$ players by using $\frac{n}{2}$ instances of the two-player version.
\end{proof}

In contrast to the fact that one can easily construct an equilibrium in $n$ steps by inserting players one after the other,
the situation is different when starting in an arbitrary situation.
We now study the dynamic that emerges if players repeatedly perform strategy changes that improve their utilities.
This may also lead to situations in which a resource is simultaneously newly allocated by two tasks of different players which necessitates the existence of a tie-breaking rule.
 We introduce two tie-breaking rules which guarantee that the game still converges towards an equilibrium.
 For an ordered budget game $\mathcal{B}$, let $p: \mathcal{N} \rightarrow \mathbb{N}$ be an injective function that assigns a unique priority to every player $i$.
 Whenever simultaneous strategy changes occur, they are executed sequentially, in decreasing order of the priorities of the players involved.
 This corresponds to setting $t_1 \prec_r t_2$ for all resources $r$ and all pairs of tasks where the priority of the player with $t_1$ was higher than the priority of the player with $t_2$.
 For $p_{\text{fix}}$, the priorities are fixed.
 For $p_{\text{max}}$, they change over time, with $p_{\text{max}}(i_0) > p_{\text{max}}(i_1)$ if $u_{i_0}(s) > u_{i_1}(s)$ for the current strategy profile $s$.
 Any ties may be broken arbitrarily.
\begin{theorem}
 Let $\mathcal{B}$ be an ordered budget game which allows multiple simultaneous strategy changes. If $\mathcal{B}$ uses either $p_{\text{fix}}$ or $p_{\text{max}}$ to set the priorities of the players, then it reaches a NE after finitely many improvement steps.
\end{theorem}
\begin{proof}
 Let $s$ be the current strategy profile $\mathcal{B}$ and $\overrightarrow{u}(s) \in \mathbb{R}^n_{\geq 0}$ the vector containing the current utilities of all players.
 We call $\overrightarrow{u}(s)$ the utility vector of $\mathcal{B}$ under $s$.
 We always sort $\overrightarrow{u}(s)$ in decreasing order of the player priorities, i.\,e. the player at position $i$ has a higher priority than the player at position $i+1$.
 For $p_{\text{max}}$, this order may change over time. Let $N \subseteq \mathcal{N}$ be the set of players who are simultaneously performing a strategy change.
 Each player would improve his utility if he were the only player in $N$. Let $s'$ be the resulting strategy profile.
 Note that $\overrightarrow{u}(s) <_{\text{lex}} \overrightarrow{u}(s')$ for both priority functions, where $<_{\text{lex}}$ is the lexicographical order.
 Let $i \in N$ be the player with the highest priority among those in $N$. For $p_{\text{fix}}$, $i$ receives exactly the utility increase he expected from the strategy change.
 From all the players in $N$, he is also the one with the smallest index in both $\overrightarrow{u}(s)$ and $\overrightarrow{u}(s')$.
 This alone warrants that $\overrightarrow{u}(s) <_{\text{lex}} \overrightarrow{u}(s')$.
 For $p_{\text{max}}$, the same argumentation holds if the position of $i$ in the utility vectors does not change.
 Otherwise, his index in $\overrightarrow{u}(s)$ is now occupied by a player $i'$ with $u_i(s) < u_i(s') < u_{i'}(s')$.
 Again, we have $\overrightarrow{u}(s) <_{\text{lex}} \overrightarrow{u}(s')$.
 Since the utility vectors are strongly monotonely increasing, but bounded by the vectors containing the maximal utility of each player, a NE is reached after finitely many steps.
\end{proof}

For the following, we assume that $p_\text{fix}$ is used as tie-breaking rule and that the priority of a player corresponds to her index.
We show that the number of improvement steps towards an equilibrium may be exponential in the number of players, even if the number of strategies per player is constant.

\begin{theorem}
 For any $n$, there is an ordered budget game $\mathcal{B}_n$ with polynomial description length in $n$ and a strategy profile $s_0$ so that the number of best-response improvement steps from $s_0$ to any NE $s$ of $\mathcal{B}_n$ is exponential in $n$.
\end{theorem}

 \begin{figure}[!ht]
  \centering
  \subfigure[extension of $\mathcal{B}_{n-1}$ that is necessary to reset $\mathcal{B}_{n-1}$]{
  \includegraphics[width=.48\textwidth]{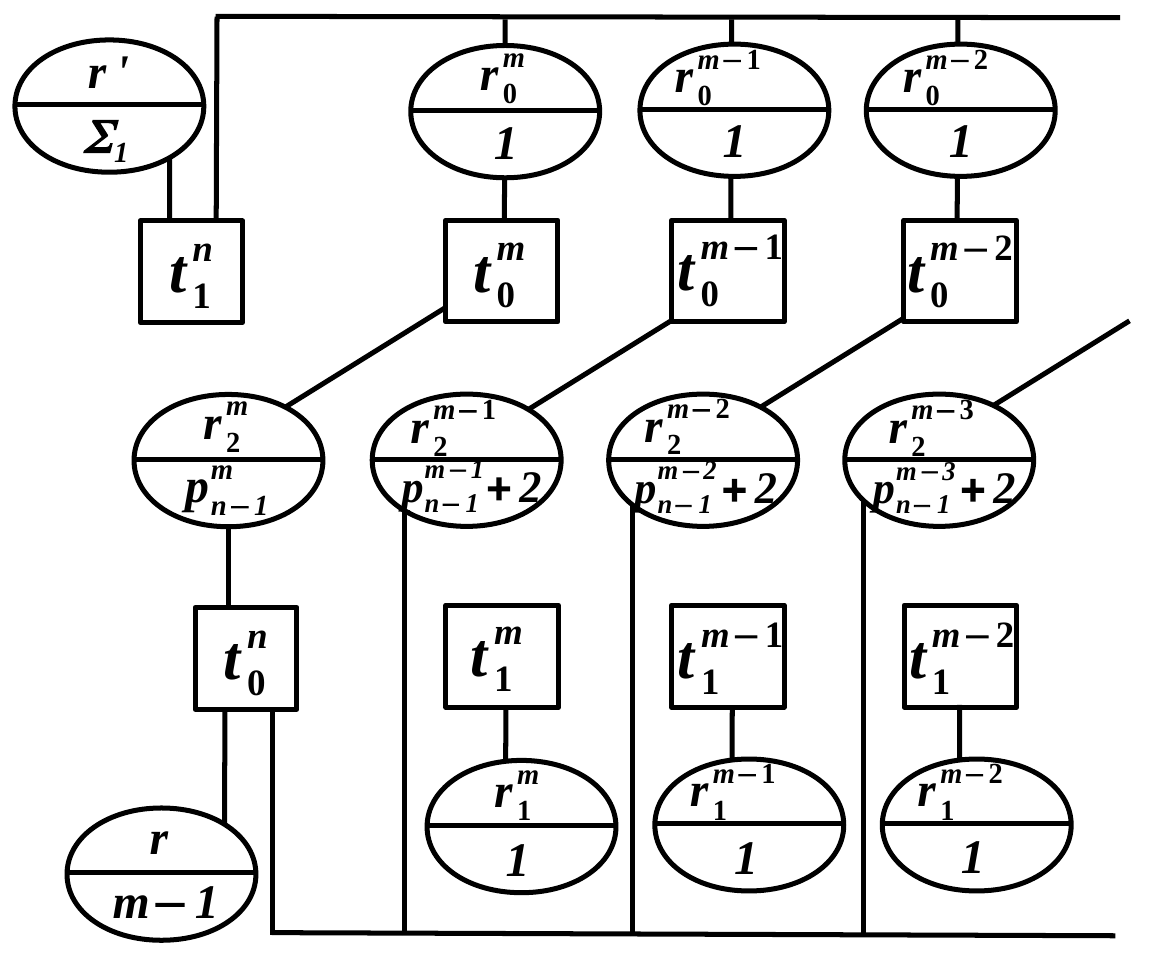}\label{fig:proofExponentialManySteps:a}}
  \subfigure[extension of $\mathcal{B}_{n-1}$ that is necessary to restart $\mathcal{B}_{n-1}$]{
  \includegraphics[width=.48\textwidth]{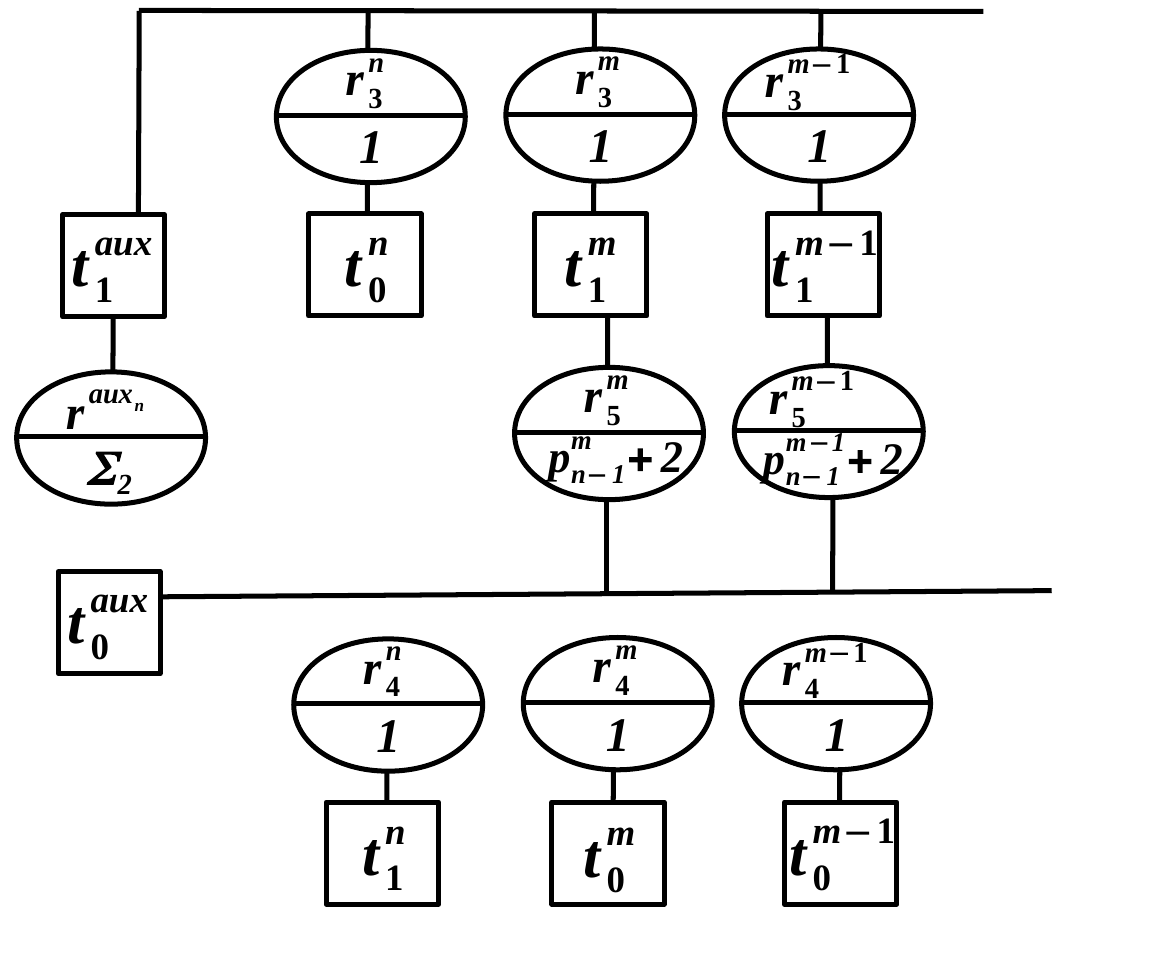}\label{fig:proofExponentialManySteps:b}}
  \caption{Construction of the ordered budget game $\mathcal{B}_n$. Figure \ref{fig:proofExponentialManySteps:a} shows the extension of $\mathcal{B}_{n-1}$ necessary to reset $\mathcal{B}_{n-1}$.
  Once all players $i$ are playing $\{t^i_1\}$, the player $n$ changes its strategy, creating for all others the incentive to go back to $\{t^i_0\}$. Here, we set $\Sigma_1 = \sum_{i=1}^{m} p^i_{n-1} + 2m$.  
  Figure \ref{fig:proofExponentialManySteps:b} shows the extensions of $\mathcal{B}_{n-1}$ which restart $\mathcal{B}_{n-1}$.
  The principle is the same, this time the player $aux$ is used to \textit{create} the new budget available.
  We set $\Sigma_2 = \sum_{i=1}^m (p^i_{n-1}) + m$. Legend in Figure \ref{fig:NPhardNE}.}
  \label{fig:proofExponentialManySteps}  
\end{figure}

\begin{proof}
 We give a recursive construction of the game $\mathcal{B}_n$.
 $\mathcal{B}_n$ contains the game $\mathcal{B}_{n-1}$,
 for which there is exactly one path of best-response improvement steps of length $\mathcal{O}(2^{n-1})$.
 $\mathcal{B}_{n-1}$ is executed once.
 Then it is reset to its original state and executed once more along the same path.
 In the end, $\mathcal{B}_n$ has reached a NE after $\mathcal{O}(2^n)$ steps.
 Each player has only two tasks and as a strategy, she can choose one of them, i.\,e. $\mathcal{S}_i = \{\{t^i_1\},\{t^i_2\}\}$.
 Labeling the strategies of player $i$ with $0_i$ and $1_i$, each strategy profile can be written as a binary number.
 The initial strategy profile $s_0$ can be regarded as 0 and the first execution of $\mathcal{B}_{n-1}$ counts up to $2^{n-1}-1$.
 The reset of $\mathcal{B}_{n-1}$ corresponds to increasing that value by 1 to $2^{n-1}$ and the second iteration of $\mathcal{B}_{n-1}$ continues counting up to $2^n-1$.
 In the final state of $\mathcal{B}_n$, every player $i$ plays strategy $1_i$.
 Since the strategies contain only single tasks, the ordering of these tasks on the resources is also an ordering of the players
 and we can abuse notation and say $i_1 \prec i_2$ for players $i_1$ and $i_2$ if $t_1 \prec_r t_2$ holds for any pair of tasks $t_1 \in \mathcal{T}_{i_1}$ and $t_2 \in \mathcal{T}_{i_2}$ and any resource $r$.
 
 In the following construction, for any pair of task $t$ and resource $r$,
 $t$ is either fully connected to $r$ or not connected to $r$ at all.
 Thus, $t(r)$ is either $b_r$ or $0$.
 In the following, \textit{connecting} a task $t$ to a resource $r$ means setting $t(r):=b_r$.
 
 We need a few new notations for our proof.
 The only NE that is reached in our construction is the state where every player $i$ plays strategy $1_i$
 and in which the players reach their final state in descending order, i.\,e. $i_1 \prec i_2$ for $i_1 > i_2$.
 Let $p^i_n$ be the utility of player $i$ in that NE for $\mathcal{B}_n$.
 For the ordered budget game $\mathcal{B}_n$, let $s^{n}_0$ be the initial strategy profile in which it is started.
 Let $i_1 \prec i_2$ for $i_1 > i_2$ in $s^{n}_0$.
 Intuitively, this means that players with a higher index get prioritized.
 
 
 For $n=1$, we build an instance $\mathcal{B}_1$ with a single player:
 $\mathcal{N} = \{1\}$ with tasks $\mathcal{T}_1 = \{t^1_1,t^1_2\}$, strategy space $\mathcal{S}_1 = \{\{t^1_0\},\{t^1_1\}\}$ and two resources $\mathcal{R} = \{r_1,r_2\}$ with $b_1 = 1,$ $b_2 = 2$.
 We connect $t^1_1$ to $r_1$ and $t^1_2$ to $r_2$.
 The initial strategy profile is $s^1_0 = (\{t^1_1\})$ and after one improvement step, $\mathcal{B}_1$ is in an equilibrium.
 
 For $n>1$, we extend the game $\mathcal{B}_{n-1}$.
 Let $m$ denote the number of players in $\mathcal{B}_{n-1}$. 
 We split the rest of the proof in two parts. 
 First, we explore how to reset $\mathcal{B}_{n-1}$ to $s^{n-1}_0$.
 The structure is sketched in Figure \ref{fig:proofExponentialManySteps:a}.
 We introduce a new player $n$ with $\mathcal{T}_n = \{t^n_0,t^n_1\}$ and $\mathcal{S}_n = \{\{t^n_0\},\{t^n_1\}\}$.
 The initial strategy of each $i$ is $\{t^i_0\}$. We now add several new resources.
 \begin{center}
 \begin{tabular}{r|c|c|c|c|c}
 Resource \  & \ $r^1_0,\ldots,r^{m}_0$ \ & \ $r^1_1,\ldots,r^{m}_1$ \ & \ $r^i_2, i = 1,\ldots,m$ \ & \ $r$ \ & $r'$ \\ \hline
 Budget \  & 1 & 1 & $p^i_{n-1} + 2$ & \ $m-1$ \ & \ $ \sum_{i=1}^{m} p^i_{n-1} + 2m$
 \end{tabular}
 \end{center}
 For $i = 1,\ldots,m$, we connect $r^i_0$ to task $t^i_0$ and $r^i_1$ to task $t^i_1$.
 Since all budgets are 1, this does not influence the game $\mathcal{B}_{n-1}$.
 We connect all $r^i_0$ to $t^n_1$ and $r$ to $t^n_0$.
 Now, the initial utility of player $n$ is $u_n(s^n_0) = m - 1$ and when all other players $i$ play strategy $\{t^i_1\}$, player $n$ can improve her utility by 1 by switching to $\{t^n_1\}$.
 
 It remains to extend the current game such that once player $n$ has switched to $t^n_1$, the remaining players $m,m-1,\ldots,1$ also switch their strategy to recreate $s^{n-1}_0$.
 For $i = 1,\ldots,m$, connect each resource $r^i_2$ to $t^n_0$ and $t^i_0$.
 This increases the utility of $t^n_0$ by $\sum_{i=1}^{m} p^i_{n-1} + 2m$.
 As a compensation, we connect the resource $r'$ to $t^n_1$.
 When $n$ switches to $t^n_1$, all the budgets of the resources $r^i_2$ become available again.
 This will cause the players $1,\ldots,m$ to change their strategies to $t^i_0$.
 Before switching, the utility of player $i$ is $p^i_{n-1} + 1$ due to the connection between $r^i_1$ and $t^i_1$.
 Switching the strategy improves this value by at least 1.
 By definition of strategy changes of a coalition with $p_\text{fix}$, we have $i_1 \prec i_2$ for all $i_1 > i_2$ and thus the resulting strategy profile is identical to the initial one for players $0, \ldots, m$.
 
 To restart the game $\mathcal{B}_{n-1}$, we apply a similar trick as before.
 The construction is sketched in Figure \ref{fig:proofExponentialManySteps:b}.
 We introduce an auxiliary player $aux_n$ with $\mathcal{T}_{aux_n} =\{t^{aux_n}_0,t^{aux_n}_1\}$, $\mathcal{S}_{aux_n} = \{\{t^{aux_n}_0\},\{t^{aux_n}_1\}\}$ and the following resources.
 \begin{center}
 \begin{tabular}{r|c|c|c|c}
 Resource \ &  \ $r^1_3,\ldots,r^m_3, r^n_3$  \ &  \ $r^1_4,\ldots,r^m_4, r^n_4$  \ &  \ $r^i_5, i = 1,\ldots,m$ \  & $r^{aux_n}$ \\ \hline
 Budget \ &  \ 1  \ &  \ 1  \ &  \ $p^i_{n-1} + 2$ \  & \ $\sum_{i=1}^m (p^i_{n-1}) + m$  \  \\
 \end{tabular}
 \end{center}
 In $s^n_0$, we set $aux_n \prec n$ and initially, her strategy is $\{t^{aux_n}_0\}$.
 We connect $t^{aux}_0$ to $r^i_5$ for all $i\in\{1,\ldots,m\}$.
 Now, this auxiliary player starts the game with a utility of $\sum_{i=1}^m p^i_{n-1}+2m$.
 We also connect $t^{aux}_1$ to $r^{aux_n}$ and to all resources $r^i_3$ for $i\in\{1,\ldots,m,n\}$.
 Finally, for every player $i=1,\ldots,m$, we connect $r^i_4$ to $t^i_0$, $r^i_4$ to $t^i_1$ and $r^i_5$ to $t^i_1$.
 For player $n$, we establish these connections the other way around, such that $r^n_3$ is connected to $t^n_0$ and $r^n_4$ to $t^n_1$.
 
 Again, the effects of $r^i_3$ and $r^i_4$ regarding $\mathcal{B}_{n-1}$ cancel out.
 Only when every player $i$ in $\mathcal{B}_{n-1}$ plays strategy $\{t^i_0\}$ and player $n$ plays strategy $\{t^n_1\}$,
 the auxiliary player will change to $t^{aux}_1$ and obtain a utility of $\sum_{i=1}^m (p^i_{n-1}) + 2m + 1$.
 This frees the budget of all resources $r^i_5$ and the utility of every task $t^i_1$ in $\mathcal{B}_{n-1}$ is increased by the same amount we increased the utility of $t^i_0$ in the first part of the construction.
 The game $\mathcal{B}_{n-1}$ is executed once more, only the player $n$ remains idle.
 When all players $i$ are playing strategy $\{t^i_1\}$, $\mathcal{B}_n$ has reached a NE.
 
 Thus, together with the auxiliary players we get a total of $2\cdot n - 1$ players in $\mathcal{B}_n$.
 and the number of steps to reach the NE is at least $2^n-1$.
 At the same time, the number of tasks and resources is polynomial in $n$.
\end{proof}

\bibliographystyle{plain}
\bibliography{references}

\end{document}